\documentclass{ifacconf}
\usepackage{natbib}
\usepackage{amsmath,amsfonts,amssymb}
\usepackage{graphicx,epsfig,psfrag,remark}
\usepackage{xfrac,bm}
\usepackage[all]{xy}
\usepackage{varioref}
\usepackage{wrapfig}
\usepackage{threeparttable}
\usepackage{dcolumn}
\newcolumntype{d}{D{.}{.}{-1}}
\usepackage{nomencl}
\makeglossary
\usepackage{comment}
\usepackage{xfrac}
\usepackage{mathrsfs}
\usepackage{float}
\usepackage[labelformat=simple]{subcaption}

\usepackage{pgfplots}
\pgfplotsset{compat=newest}
\usepackage{url}
\usepackage{tikz}
\usetikzlibrary{shapes.geometric}
\usepackage{color}
\usepackage{multicol}
\usepackage{algorithm}
\usepackage[noend]{algpseudocode}

\newenvironment{proof}{\paragraph*{Proof:}}{\hfill$\square$}

\usetikzlibrary{external}
\tikzsetexternalprefix{Figures/}

\newcommand\numberthis{\addtocounter{equation}{1}\tag{\theequation}}

\newcommand{\ud}{\mathrm{d}}

\newcommand{\f}{{\mathsf{f}}}
\newcommand{\sa}{{\mathsf{s}}}

\newcommand{\cA}{\mathcal{A}}

\newcommand{\cJ}{\mathcal{J}}

\newcommand{\cL}{\mathcal{L}}

\newcommand{\cN}{\mathcal{N}}

\newcommand{\cR}{\mathfrak{R}}
\newcommand{\cS}{\mathcal{S}}

\newcommand{\cT}{\mathcal{T}}
\newcommand{\cU}{\mathcal{U}}

\newtheorem{problem}{Problem}
\newtheorem{theorem}{Theorem}

\newcommand{\norm}[1]{\left\lVert#1\right\rVert}

\renewcommand{\baselinestretch}{0.95}
\newremark{remark}{Remark}

\begin{document}
\begin{frontmatter}

\title{Greedy Decentralized Auction-based Task Allocation for Multi-Agent Systems \thanksref{footnoteinfo}}

\thanks[footnoteinfo]{This work was supported in part by a Fellowship of the Belgian American Educational Foundation.}

\author[First]{Martin Braquet} 
\author[First]{Efstathios Bakolas} 

\address[First]{Department of Aerospace Engineering and Engineering Mechanics, The University of Texas at Austin, Austin, Texas 78712-1221, USA (e-mail: braquet@utexas.edu, bakolas@austin.utexas.edu).}

\begin{abstract}
We propose a decentralized auction-based algorithm for the solution of dynamic task allocation problems for spatially distributed multi-agent systems. In our approach, each member of the multi-agent team is assigned to at most one task from a set of spatially distributed tasks, while several agents can be allocated to the same task. The task assignment is dynamic since it is updated at discrete time stages (iterations) to account for the current states of the agents as the latter move towards the tasks assigned to them at the previous stage. Our proposed methods can find applications in problems of resource allocation by intelligent machines such as the delivery of packages by a fleet of unmanned or semi-autonomous aerial vehicles. In our approach, the task allocation accounts for both the cost incurred by the agents for the completion of their assigned tasks (e.g., energy or fuel consumption) and the rewards earned for their completion (which may reflect, for instance, the agents' satisfaction). We propose a Greedy Coalition Auction Algorithm (GCAA) in which the agents possess bid vectors representing their best evaluations of the task utilities. The agents propose bids, deduce an allocation based on their bid vectors and update them after each iteration. The solution estimate of the proposed task allocation algorithm converges after a finite number of iterations which cannot exceed the number of agents. Finally, we use numerical simulations to illustrate the effectiveness of the proposed task allocation algorithm (in terms of performance and computation time) in several scenarios involving multiple agents and tasks distributed over a spatial 2D domain.
\end{abstract}


\end{frontmatter}

\maketitle
\thispagestyle{plain}
\pagestyle{plain}

\section{Introduction}\label{s:intro}

We present a decentralized auction-based algorithm to address dynamic task allocation problems for multi-agent systems. In our problem, the agents have to complete a set of tasks which are distributed over a given spatial domain. We propose a decentralized solution for the computation of task assignment profiles based on auction-based negotiations between the agents.
Our proposed methods can find applications in problems in which agents (e.g., autonomous vehicles, humans, robots, intelligent machines, etc.) have to share resources and distribute the workload among them in order to accomplish one or more tasks. 
Disaster response by a fleet of unmanned aerial vehicles (UAV) which have to assess the severity of the situation and discover where help is needed more as well as the delivery of packages by autonomous or semi-autonomous ground or aerial robots are two characteristic examples.

\textit{Literature review:} There are several types of task allocation problems for multi-agent systems depending on the ability of each agent to handle multiple tasks (involving task scheduling) and on whether it is possible to have multiple agents assigned to the same task (thus, allowing for the formation of coalition of agents). These problems can be addressed by auction-based techniques, distributed and / or multi-objective optimization, game-theoretic methods and machine-learning algorithms, to name but a few.

An important consideration when developing algorithms for multi-agent task allocation is the ability of these algorithms to be deployed in systems where there is no single entity that allocates tasks and workload among the agents. In this regard, centralized methods rely on a single point of operation in the sense that the agents negotiate with each other under the direction of a central entity (\cite{p:gerkey2002}). 
Decentralized methods avoid this single point of failure by allowing each agent to consult directly with the other agents and compute their own task assignments. Decentralized execution, however, adds significant computation time (\cite{p:howauction2009,p:NANJANATH2010,p:capitan2013}). 

Auction-based approaches are derived from market economy principles in which each agent tries to maximize his own profit, based on the total reward that will then be redistributed among them. These methods are receiving an increased amount of attention (e.g., satellites in \cite{p:phillips2021}, drones in \cite{p:hayat2020}) mainly because of certain key benefits such as the worst-case global utility that can be derived theoretically by using them (\cite{p:guannan2019}), their fast convergence, low complexity and high computational efficiency (\cite{p:keum2019,p:shin2019}). Auction-based methods have also been boosted by recent breakthroughs in reinforcement learning (\cite{p:rahili2020}).
The consensus-based bundle algorithm (\cite{p:howauction2009}) (CBBA) utilizes a market-based decision strategy as the mechanism for decentralized task selection and uses a consensus routine based on local communication as a conflict resolution mechanism to achieve agreement on the winning bid values.
Finally, other decentralized auction algorithms based on local communication have been developed to allow the agents to bid on a task asynchronously (\cite{p:johnson2011}).

One of the simplest approaches to solve decentralized auction-based problems is via greedy algorithms, which consider the optimal (in a myopic sense) choice that maximizes a global objective (\cite{p:luo2012}).
Some approaches handle heterogeneous agents (with different traits / capabilities) by computing their utilities based on their own local information, and the task allocation is solely determined by their bids (\cite{p:ravichandar2019}).
In this regard, such algorithms can calculate the agents' utilities based on their resource levels and the possibility of visiting refill stations (\cite{p:lee2018}).
Auction-based techniques have been proven to efficiently produce suboptimal solutions (\cite{p:gerkey2004}) with a guaranteed convergence to a conflict-free assignment. Other advantages of auctions are their high scalability and robustness to variations in the communication network topology (\cite{p:whitbrook2019, p:otte2020}).

Other types of task allocation methods include those which are based on game theory and in particular, potential games for the computation of mutually agreeable task assignments. Although the negotiation protocols are proven to converge to mutually agreeable tasks (\cite{p:shamma2007}), their convergence is only guaranteed for the case in which the game remains the same (task utilities are constant) which is not the case in a dynamic task allocation problem. 
Other algorithms aim to compute a mutually agreeable profile corresponding to a Nash equilibrium (game-theoretic formulation of task allocation problems) for all agents (\cite{p:bakolas2020}). Game theory is an important tool to extend task allocation problems to multiple agents but finding \textit{efficient} Nash equilibria (task assignment profiles giving high global utility) is not guaranteed (solutions based on individual rationality may not automatically lead to high global utility) and computational cost can be significant.
Likewise, constrained optimization approaches based on nonlinear programming tools require in general significant computational power and time. More efficient optimization-based task allocation methods that rely on tools from quadratic programming have been introduced in (\cite{p:bakshi2019quad}).
Recently, machine-learning algorithms have started to receive a significant amount of attention mainly because they can process a lot of information (by utilizing, for instance, neural networks) and handle unknown environments via reinforcement learning (especially Deep Q-Learning~\cite{p:gautier2020}). Recurrent neural networks also find applications in scheduling problems for clustered tasks in Multi-Task Robots Single-Robot Tasks Assignment problems (\cite{p:bakshi2019}).

\textit{Contributions:} In this paper, we propose a dynamic auction-based task allocation algorithm. In our approach, the task utilities depend on both the rewards earned by the agents for accomplishing their assigned tasks as well as the costs they incur while doing so (the latter correspond to cost-to-go functions of relevant optimal control problems). The utilities are thus in general dependent on the state of the agents. In this context, the agents can only perform one task while several agents can be assigned the same task (if this is beneficial to them and their team). In contrast with game-theoretic algorithms which may not always achieve high global utility for the team (inefficient Nash equilibria), our proposed auction-based task allocation mechanism finds task assignments that maximize the global utility of the system. A key advantage of our proposed approach is time efficiency, yet with reasonably high global utility.

We propose a Greedy Coalition Auction Algorithm (GCAA) where the agents negotiate while moving in their state space towards their assigned tasks. When an agent changes his assignment, he needs to recompute the cost estimate and thus his own state-dependent utility. In contrast to game-theoretic solutions (\cite{p:bakolas2020}) which aim for individual rationality but cannot guarantee good team performance, we do not seek a mutually agreeable task assignment but consider instead a broader set of solutions that allows for a higher global utility. Furthermore, in contrast with the CBBA algorithm which clusters and schedules a sequence of tasks for each agent, in this work the problem is composed of multiple agents making a coalition for a specific task that is spatially distributed (which is the only task for that agent). This work hence falls under the category of \textit{Single-Task Robots Multi-Robot Tasks Instantaneous Assignment} (ST-MR-IA) problem, also known as the coalition
formation problem (\cite{p:gerkey2004}).

\textit{Outline:} The rest of the paper is presented as follows. We discuss the problem setup in Section~\ref{s:prelim}. In Section~\ref{s:utility}, we identify the utilities for the tasks, the agents individually, and the team as a whole.
The proposed dynamic auction-based task allocation algorithm and the theoretical analysis on its convergence are presented in Section~\ref{s:RHTA}. In Section~\ref{s:simu}, we present extensive numerical simulations. Finally, concluding remarks and directions for future work are provided in Section~\ref{s:concl}.

\section{Problem Setup}\label{s:prelim}

We assume a multi-agent system (MAS) comprised of $n$ agents. These agents are called active agents when they are far from their target so that they can recompute their best task assignment while moving toward the target, otherwise they are called passive agents when they are too close to the target to consider other targets (they are then permanently assigned to this final task). Let $x_i \in \cS_i \in \Sigma$ and $u_i \in U_i$, for $i \in [1,n]_d$ be the state and input of the $i$-th agent of the MAS at time $t\geq 0$ ($\cS_i$ being his state space and $U_i$ his input space), and $\Sigma \subseteq \mathbb{R}^m$. We also define $\bm{x} \in \cS$ the joint state of the MAS, in which $\bm{x} := (x_1,\dots,x_n)$ and $\cS:= \cS_1 \times \dots \times \cS_n$ (joint state space). Let $\bm{u} \in U$ be the joint input of the MAS, where $\bm{u} := (u_1, \dots, u_n)$ and $U:= U_1 \times \dots \times U_n$ (joint input space). 

The motion of the $i$-th agent is described by
\begin{equation}\label{eq:motion1}
\dot{x}_i = f_i(\bm{x}, \bm{u}),~\quad~x_i(0)=x_i^0,~\quad~i \in [1, n]_d,
\end{equation}
where $x^0_i \in \cS_i$ is the initial state of the $i$-th agent and $f_i: \cS_i \times U_i \rightarrow \cS_i$ is his associated vector field. 
Consequently, $\bm{x}^0 = (x^0_1, \dots, x^0_n) \in \cS$ is denoted as the joint initial state.

In general, task allocation aims to assign individual tasks for $n$ agents and $p$ tasks, $\cT := \{ \cT_1, \dots, \cT_p\}$. 
Let $X_{\cT}$ be the set of states associated with the given tasks, where $X_{\cT} := \{ x_{\cT_1}, \dots, x_{\cT_p} \}$,
and $\cA_i:= \{ a_i^k:~k\in[1,\mathrm{card}( \cA_i )]_d \}$ the set of possible task assignments for the $i$-th agent given a set of tasks $\cT$. While the agents have limited communication between each other, we suppose that they have complete information about all the tasks available.
Each assignment $a_i^k \in \cA_i$ is equal to either a task in $\cT$, that is, $a_i^k = \cT_\ell$ where $\cT_\ell \in \cT$, or the null assignment, that is, $a_i^k = a_\varnothing$.

Additionally, we denote the set of active agents as $\cN_a \subseteq [1,n]_d$ and we fix the assignment $a_i$ of agent $i$ (thus switching his status from active to passive) for all $t > t_p$ if the agent lies inside the boundary of the target, that is, $\Phi_i( x_i(t_p), x_{\cT_j}) < 0$ where $\Phi_i( x_i(t_p), x_{\cT_j})$ is a boundary constraint; for instance $\Phi_i( x_i(t_p), x_{\cT_j}) := \|x_i(t_p) - x_{\cT_j} \| - R_p$ where $R_p$ is the minimum agent-to-target distance to make the task assignment permanent.

\section{Task Utilities}\label{s:utility}

The task utility is characterized by a reward obtained for the completion of the task $\cT_j \in \cT$ and a state-dependent cost to finish this task (for example, the transition cost due to the motion of the agent).

\noindent \textit{Static task utility:} Given an assignment profile $\bm{a} = (a_1, \dots, a_n)$, we denote by $\cT^{-1}_j( \bm{a} )$ the index-set corresponding to the agents assigned to task $\cT_j \in \cT$ under the particular profile. Since a task is not necessarily accomplished when an agent is assigned to it, we let $p_{ij} \in [0,1]$ be the probability of the task $\cT_j$ to be completed successfully by the $i$-th agent. In this case, the probability that the task is successfully completed by at least one agent increases with the number of agents assigned to this task.
The expected reward for completing task $\cT_j$ is defined as~\cite{p:bakolas2020}:
\begin{equation}
r_{\cT_j}(\bm{a}) = \bar{r}_{\cT_j}\left[1-\prod\nolimits_{i \in \mathcal{T}_j^{-1} (\bm{a})} (1-p_{ij})\right], \label{eq}
\end{equation}
where $\bar{r}_{\cT_j}$ is the nominal reward of $\cT_j$. Indeed, the probability that at least one agent completes the task is equal to the complementary of the probability that no agent completes the task, i.e. $\prod\nolimits_{i \in \mathcal{T}_j^{-1} (\bm{a})} (1-p_{ij})$.
It is worth noting that the assignments (and their associated utility) of the passive agents are also taken into account to compute the total reward.

\noindent \textit{State-dependent task completion cost:} The cost to finish the task $\cT_j$ associated with the state $x_{\cT_j}$ at time $t=t_{\f,\cT_j}$ by the $i$-th agent is defined as the optimal cost related to the optimal control problem presented in Problem~\ref{problemOCP}.
\begin{problem}\label{problemOCP}
Let $a_i = \cT_j$, where $\cT_j \in \cT$ and $i \in [1,n]_d$. Furthermore, we denote $x_{\cT_j} \in \cS_i$ as the state linked to $\cT_j$ and $t_{\f,\cT_j}>0$ as the related completion time for $\cT_j$. The goal is to obtain an optimal control input $u_i^{\star}(\cdot): [0,t_\f] \rightarrow U_i$ that is piece-wise continuous and minimizes the functional given by:
\begin{align}\label{eq:cost}
\cJ_i(u_i(\cdot); x_i^0, x_{\cT_j}) 
:= \int_0^{t_\f} \cL_i(x_i(t), u_i(t)) \ud t,
\end{align}
such that the dynamic constraints \eqref{eq:motion1} and the following terminal constraint $\Psi_i( x_i(t_\f), x_{\cT_j} ) = 0$, where $\Psi_i(\cdot; x_{\cT_j})$ is a given $C^1$ function, are respected. Finally, the minimum cost is given by $\rho_i(x_i^0; x_{\cT_j}) := \cJ_i(u_i^{\star}(\cdot);x_i^0, x_{\cT_j})$. 
\end{problem}
\begin{remark}
The terminal constraint function $\Psi_i$ follows $\Psi_i( x_i(t_\f), x_{\cT_j} ) = x_i(t_\f) - x_{\cT_j}$, implying $x_i(t_{\f,\cT_j}) = x_{\cT_j}$. 
Also, a second type of terminal constraint $\Psi_i( x_i(t_\f), x_{\cT_j} ) = \norm{x_i(t_\f) - x_{\cT_j}} - R_{\cT_j}$ considered in this work requires some agents to loiter around the target $\cT_j$ with a certain radius $R_{\cT_j}$
during a loitering time $\tau_{\cT_j} \in [0, t_{\f,\cT_j}]$, in which case these agents start loitering at time $t_{\f,\cT_j} - \tau_{\cT_j}$.
\end{remark}

\noindent \textit{Total Task Utility:} The total completion cost of task $\cT_j$ given the assignment profile $\bm{a} = (a_1, \dots, a_n)$ is given by 
\begin{equation}
\cR_{\cT_j}(\bm{a};\bm{x}^0, x_{\cT_j}) := \sum\nolimits_{i \in \cT^{-1}_j(\bm{a}) } \rho_i(x_i^0; x_{\cT_j}),
\end{equation}
which leads to the definition of the total task utility associated with task $\cT_j$ for a given $\bm{x}_0$
\begin{equation}
\cU_{\cT_j}(\bm{a};\bm{x}^0) := r_{\cT_j}(\bm{a}) - \lambda_{\cT_j} \, \cR_{\cT_j}(\bm{a},\bm{x}^0; x_{\cT_j})
\end{equation}
where $\lambda_{\cT_j}$ is a constant which is used to convert the cost-to-go to the same units as the reward (e.g. from a loss of energy to a loss of money).

\noindent \textit{Individual, Team Utilities:} 
Let us denote the global utility
as
\begin{equation}\label{eq:teamutility}
\cU(\bm{a};\bm{x}^0) := \sum\nolimits_{\cT_j\in \cT} \cU_{\cT_j}(\bm{a};\bm{x}^0).
\end{equation}
The goal is to set this team's utility equal to the sum of each individual utility in order to maximize each individual utility separately. In this regard, based on the task assignment $\bm{a}$, we set the individual utility of agent $i$ equal to his marginal contribution to the global utility $\cU(\bm{a};\bm{x}^0)$:
\begin{align*}
\cU_{i}(\bm{a};\bm{x}^0) & := \cU(\bm{a}; \bm{x}^0) - \cU((a_{\varnothing}, a_{-i}); \bm{x}^0) \numberthis \label{eq:marginaleq1} \\
& = \cU_{\cT_j}(\bm{a}; \bm{x}^0) - \cU_{\cT_j}((a_{\varnothing}, a_{-i}); \bm{x}^0).
\end{align*}

\section{Dynamic Task Allocation}\label{s:RHTA}

\subsection{Problem formulation}
The task allocation is called dynamic since the utilities of the agents change along their path towards their target (state-dependent cost and agents obtain new information by communicating with other agents in the surrounding). In this case, a new assignment profile $\bm{a}^{\star}(t)$ has to be selected at each time step $t \in [0,t_\f]$ as the agents evolve in their state space.

\begin{problem}[DTA: Dynamic Task Allocation]\label{problemRHTA}
Let $t_\f > 0$ and $\bm{x}^0 \in \cS$, find a time-varying task assignment profile $\bm{a}^{\star}(\cdot): [0,t_\f] \rightarrow \cA$ for all the remaining active agents $i \in \cN_a$, that maximizes the global utility in a decentralized way (communication constraints) according to the permanent assignment $a_{i_p}$ of the passive agents $i_p \in \cN_p = [1,n]_d \setminus \cN_a$ and the terminal constraints.
\end{problem}

\subsection{Auction protocols for decentralized task allocation}\label{ss:negotiate}
The main principle of auctions consists in the computation of agents' individual utility for some tasks. Based on these proposed bids, the agents communicate between each other in order to deduce the best allocation for each of them.
A key point is that for their realization, an agent does not have to know the utilities of his teammates (decentralized implementation). The main idea behind the algorithm is to find the best task coalition for the multi-agent network by allocating the tasks to the agents obtaining the highest utility (also called greedy approach).

\subsection{Greedy Coalition Auction Algorithm}

The GCAA is an auction-based algorithm that leverages the simplicity of greedy approaches to provide a solution with fast convergence. The main idea is to iterate between an auction phase and a consensus phase such that it converges to a winning bids list (\cite{p:howauction2009}).

Each agent has three vectors that are constantly updated at each iteration step $t$. The first vector $\bm{z}_i \in [0,p]_d^{n}$ is the list of selected tasks among $\mathcal{T}$, meaning that agent $i$ possesses a vector $\bm{z}_i$ of length $n$ where the $k$-th element of the vector is the expected task assignment of agent $k$ to the best knowledge of agent $i$. The second vector $\bm{y}_i \in \mathbb {R}_{>0}^{n}$ is the list of winning bids (agent's utilities), that is, the $k$-th element of $\bm{y}_i$ is the expected individual utility of agent $k$ by selecting the task $z_{i,k}$ ($k$-th element of the vector of selected tasks $\bm{z}_i$). The third vector $\bm{c}_i \in [0,1]_d^n$ is the list of finalized (or completed) allocations and informs agent $i$ about the status of the allocation for the other agents. In particular, the $k$-th element of $\bm{c}_i$ is set to 1 if the agent $k$ does not plan to change his target anymore, and 0 otherwise. This way, the agents for which the assignment is completed are not taken into account for the auction process in subsequent steps. Based on these three vectors that are first updated, each agent will decide and propose the best assignment for themselves (i.e., maximizing their own utility).
The main algorithm is presented in Main Algorithm and the two associated phases are explained next.

\begin{algorithm}[ht]
\textbf{Main Algorithm:} Greedy Coalition Auction Algorithm\\
\textbf{Input}: $\bm{x}^0$\\
\textbf{Output}: $\bm{z}(t)$
\begin{algorithmic}[1]
\State $t = 0$
\State $\bm{y}(0) = \bm{0}$
\State $\bm{z}(0) = \bm{0}$
\State $\bm{c}(0) = \bm{0}$
\While {$\exists i:\, c_{i,i}(t) = 0 $ \textbf{and} $c_{i,i}(t-1) = 0$ }
  \State SelectBestTask()
  \State ShareStateVectors()
  \State UpdateStateVectors()
  \State $t = t+1$
\EndWhile
\end{algorithmic}
\end{algorithm}

\subsubsection{Auction process:}

The first phase of the algorithm is the auction process.
Here, each agent aims to select his best task based on his own utility.
At lines 2--4 of Algorithm 1, the previous bid vectors are copied into the current ones. If the task selected by one agent $i$ is not finalized (line 5), agent $i$ picks the task $J_i$ that maximizes his expected utility (lines 6--7). Agent $i$ then updates his bid vector with the selected task (line 8) and the associated utility (line 9).

\begin{algorithm}[ht]
\caption{Select the best task for agent $i$}
\textbf{Function} SelectBestTask\\
\textbf{Input}: $\bm{y}(t-1), \bm{z}(t-1), \bm{c}(t-1), \bm{x}^0$ \\
\textbf{Output}: $\bm{y}(t), \bm{z}(t)$
\begin{algorithmic}[1]
\For{$i \in [1,n]_d$}
  \State $\bm{y}_i(t) = \bm{y}_i(t-1)$
  \State $\bm{z}_i(t) = \bm{z}_i(t-1)$
  \State $\bm{c}_i(t) = \bm{c}_i(t-1)$
  \If {$c_{i,i}(t) = 0$}
    \State $\bm{a} = \bm{z}_i(t)$
    \State $J_i = \mathop{{\arg\!\max}}_j \cU_{i}((z_{i,j}(t), a^{\star}_{-i});\bm{x}_i^0)$
    \State $z_{i,i}(t) = J_i$
    \State $y_{i,i}(t) = \cU_{i}(\bm{z}_{i}(t);\bm{x}^0_i)$
  \EndIf
\EndFor
\end{algorithmic}
\end{algorithm}

\subsubsection{Consensus process:}

In Algorithm 2, the consensus process first aims to share the bid vectors $\bm{y}_i$, $\bm{z}_i$, $\bm{c}_i$ with the other agents within the communication range of agent $i$. For each agent $i$, the agents $k$ within the communication range of agent $i$ (satisfying $g_{ik}(t) = 1$ at lines 1--2) send their bid vectors $y_{k,k}(t), z_{k,k}(t)$ and $c_{k,k}(t)$ (lines 3--5).
Then in Algorithm 3, based on his winner bids vector, agent $i$ determines the set of agents $\Tilde{\cA}_i(t)$ allocated to the same selected task (line 2) and extracts the winner based on their utility (line 3). He adds the winner to the list of finalized allocations $\bm{c}_i$ (line 4) and resets the values of the losers in the bids $\bm{y}_i$ and tasks $\bm{z}_i$ (lines 5--8).

\begin{algorithm}[ht]
\caption{Share the bid vectors to agent $i$}
\textbf{Function} ShareStateVectors\\
\textbf{Input}: $\bm{y}(t), \bm{z}(t), \bm{c}(t)$ \\
\textbf{Output}: $\bm{y}(t), \bm{z}(t), \bm{c}(t)$
\begin{algorithmic}[1]
\For{$i \in [1,n]_d$}
  \For{$k \in \{k\,|\,g_{ik}(t) = 1\}$}
    \State $z_{i,k}(t) = z_{k,k}(t)$
    \State $y_{i,k}(t) = y_{k,k}(t)$
    \State $c_{i,k}(t) = f_{k,k}(t)$
  \EndFor
\EndFor
\end{algorithmic}
\end{algorithm}

\begin{algorithm}[ht]
\caption{Update the bid vectors of agent $i$ according to the winners/losers}
\textbf{Function} UpdateStateVectors\\
\textbf{Input}: $\bm{y}(t), \bm{z}(t), \bm{c}(t)$ \\
\textbf{Output}: $\bm{y}(t), \bm{z}(t), \bm{c}(t)$
\begin{algorithmic}[1]
\For{$i \in [1,n]_d$}
  \State $ \Tilde{\cA}_i(t) = \{ k\, | \, z_{i,k}(t) = z_{i,i}(t), f_{i,k}(t) = 0 \}$
  \State $K_i = \mathop{{\arg\!\max}}_{k \in \Tilde{\cA}_i(t)} y_{i,k}(t)$
  \State $c_{i,K_i}(t) = 1$
  \For{$k \in \Tilde{\cA}_i(t) \setminus K_i$}
    \State $z_{i,k}(t) = 0$
    \State $y_{i,k}(t) = 0$
  \EndFor
\EndFor
\end{algorithmic}
\end{algorithm}

Then the time is updated ($t \leftarrow t+1$) and the main algorithm loops to Algorithm 1.
Finally, the algorithm has converged when the finalized choices are validated for some agents ($c_{i,i} = 1$) and the other agents not assigned to a task ($c_{i,i} = 0$) have not changed since the past iteration (meaning that the cost to reach each task is higher than the marginal reward they can obtain).

Once the task allocation is completed, the agents move according to the solution of Problem~\ref{problemOCP} minimizing the cost from the agent to the target. In order to prevent abrupt trajectory changes during the dynamic allocation, we stop the computation of the allocation when the agents are close to their associated target, that is, if $t > t_{\sa,\cT_j}$ where $t_{\sa,\cT_j}$ is the stop time for agent $\cT_j$.

\subsection{Application example}

To illustrate the main steps of the algorithm through a simple example with 2 tasks and 4 agents, Fig.~\ref{fig:task_alloc_ink} shows a task allocation along with their bid vectors. The communication links are shown with red dashed lines and the final task allocation is given with green dashed lines. 

\begin{figure}[h]
\fontsize{5pt}{11pt}
\def\svgwidth{\linewidth}
\centering
\input{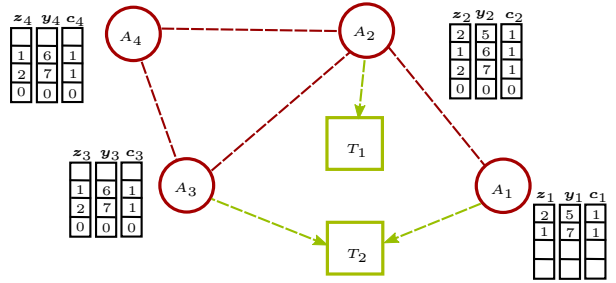}
\caption{Graphical illustration of the auction-based greedy algorithm.}
\label{fig:task_alloc_ink}
\end{figure}

Since agents $A_1$ and $A_3$ cannot communicate directly with each other, they assume that they will obtain the entire reward by completing their selected task $T_2$ while they will actually need to split it. In more details, each agent $i$ fills in his bid vector (associated to his $i$-th row) depending on his best assignment during the first iteration. For example, agent $A_1$ chooses task $T_2$ with a utility of 5 while agent $A_2$ chooses task $T_1$ with a utility of 6. Then, they share their bid vector (i.e., fill in their rows) with their neighbors only, so that $A_1$ does not have information about $A_3$, and reciprocally. Each agent finally updates his bid vector by selecting the task with the highest utility and setting the associated assignment status $\bm{c}$ to 1 (e.g., at the first iteration, all agents finalize the assignment of $A_2$ because he proposes a utility of 7). At the next iterations, the assignment of $A_2$ is no longer computed and the other agents take into account the permanent assignment of $A_2$ for the computation of their own utility (e.g., $A_4$ no longer proposes a bid for $T_2$ because the reduction of the marginal reward associated to the coalition with $A_2$ dropped his marginal utility below zero, it is thus preferable for $A_4$ not to select any task by securing a null utility). At the second iteration, $A_1$ and $A_3$ propose and finalize their assignment for $T_2$ since they think that they are completing $T_2$ individually (no communication between them) and $A_4$ does not propose any assignment. This example thus shows that communication constraints can lead to suboptimal solutions because the actual utility that $A_1$ and $A_3$ will receive by completing $T_2$ is lower than their prediction.

\begin{theorem}
Consider the auction-based task allocation process solved by the GCAA algorithm (Main Algorithm) where the communication range can be limited. Let $n$ be the number of agents, then GCAA converges to an assignment within at most $n$ steps.
\end{theorem}
\begin{proof}
The proof is derived from the definition of greedy algorithms. In particular at each time iteration $t$ and for all agents $i \in [1,n]_d$, one element (index $K_i$ as presented in Algorithm 3) of $\bm{c}_{i}$ is set to 1 as the task of agent $i$ is set to be finalized. As a consequence at time $t$, there are $t$ elements of $\bm{c}_{i}$ set to 1 and $n-t$ elements still initialized to 0. Hence at time $t=n$, all the elements of $\bm{c}_{i}$ are set to 1 for each agent $i$ which means that the stopping criteria in Main Algorithm ($c_{i,i} = 1$ for all agents $i$) is necessarily verified. The algorithm is thus proven to converge after at most $n$ steps (the number of agents).
\end{proof}

\begin{remark}
This convergence theorem guarantees that the computation time is growing linearly with the number of agents.
\end{remark}

\section{Numerical Simulations}\label{s:simu}

In this section, we present numerical simulations\footnote{Source code available at \url{https://github.com/MartinBraquet/task-allocation-auctions}.} to illustrate the main ideas of the methods proposed so far. We consider a team of agents with double integrator dynamics, that is, $\ddot{p}_i = u_i$, 
with $p_i(0)=p_i^0$ and $\dot{p}_i(0)=v_i^0$, where $p_i \in \mathbb{R}^2$ ($p_i^0 \in \mathbb{R}^2$) and $\dot{p}_i \in \mathbb{R}^2$ ($v_i^0 \in \mathbb{R}^2$) denote, respectively, the position and velocity of the $i$-th agent at time $t$ ($t_0=0$), $i \in [1,n]_d$. The performance index is given by the control effort $
\cJ(u_i(\cdot)) := (1/2) \int_{0}^{t_\f} |u_i(t)|^2 \mathrm{d}t$ and the conversion constant is $\lambda_{\cT_j} = 1$ ($j \in [1,p]_d)$. By setting $x_i := (p_i,~\dot{p}_i) \in \mathbb{R}^4$ and $x_{\cT_j} := (p_{\cT_j}, 0) \in \mathbb{R}^4$, the terminal constraint function is chosen randomly between:
\begin{itemize}
    \item $\Psi_i(x_i(t_{\f,\cT_j}); x_{\cT_j}) := x_i - x_{\cT_j}$,  which means that the $i$-th agent tries to reach the position $p_{\cT_j}$ associated with his assigned task $\cT_j$ at time $t=t_{\f,\cT_j}$ and with terminal velocity $\dot{p}_{\cT_j}$ (randomly selected). 
    \item $\Psi_i(x_i(t_{\f,\cT_j}); x_{\cT_j}) := \|p_i - p_{\cT_j}\| - \tilde{R}_{\cT_j}$,  which means that the $i$-th agent tries to reach the circle (with radius $\tilde{R}_{\cT_j}$) around his assigned task $\cT_j$ at time $t=t_{\f,\cT_j}-\tau_{\cT_j}$ and then loiters around the target until $t_{\f,\cT_j}$. In this work, the best entry point to enter the circle is selected by discretizing the circle in 10 points and selecting the point that minimizes the cost function\footnote{The best solution can also be found by optimal control methods in a systematic / rigorous way and will be considered in further work}. 
\end{itemize}

Both terminal constraints are associated with an optimal control problem with non-zero initial and terminal velocities. It turns out (see, for instance, \cite{p:battin1987}) that the optimal control input is given by 
\begin{align*}
u_i^{\star}(t;t_\f, x_i^0, x_{\cT_j}) & = \frac{4}{t_\f - t} \, \Big[\dot{p}_{\cT_j} - \dot{p}_i(t)\Big] \addtocounter{equation}{1}\tag{\theequation} \\
& + \frac{6}{(t_\f - t)^2} \, \Big[p_{\cT_j} - p_i(t) - \dot{p}_{\cT_j}\,(t_\f - t)\Big]
\end{align*}
which defines a second-order differential equation with time-varying coefficients where $u_i^{\star}(t) = \ddot{p}_i(t)$. It is solved numerically using integration tools (ODE45) in \textsc{Matlab}. 

\begin{figure}[h]
\centering
\begin{subfigure}{\linewidth}
  \centering
  \includegraphics[width=0.8\linewidth]{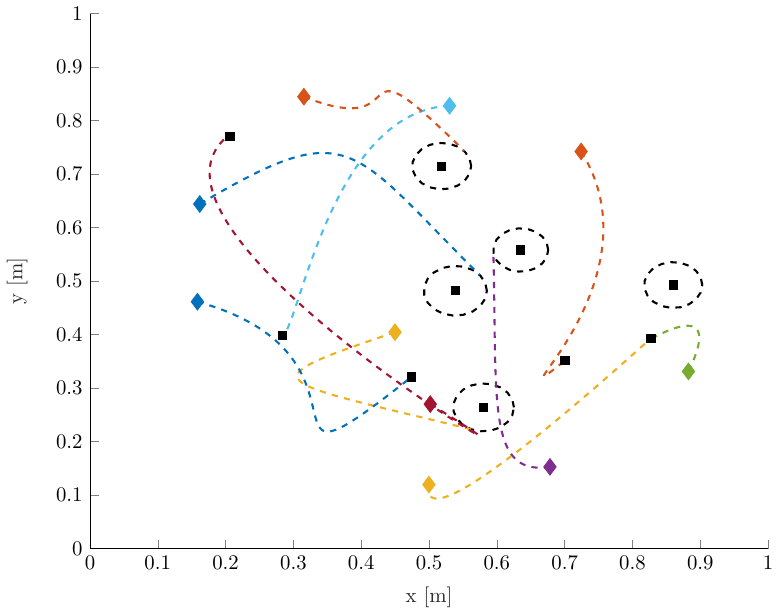}
  \caption{$t=0^+$}
  \label{fig:ExUC_Im}
\end{subfigure}
\begin{subfigure}{\linewidth}
  \centering
  \includegraphics[width=0.8\linewidth]{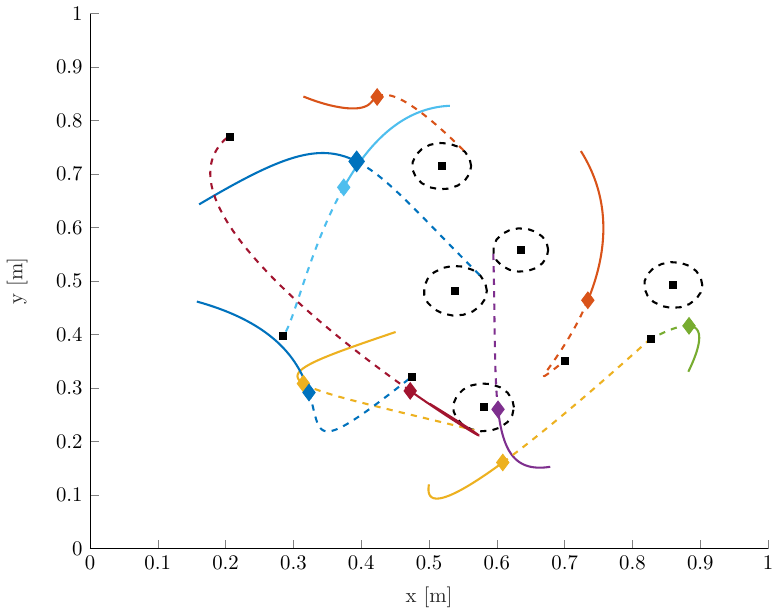}
  \caption{$t=4$}
  \label{fig:ExUC_Im2}
\end{subfigure}
\begin{subfigure}{\linewidth}
  \centering
  \includegraphics[width=0.8\linewidth]{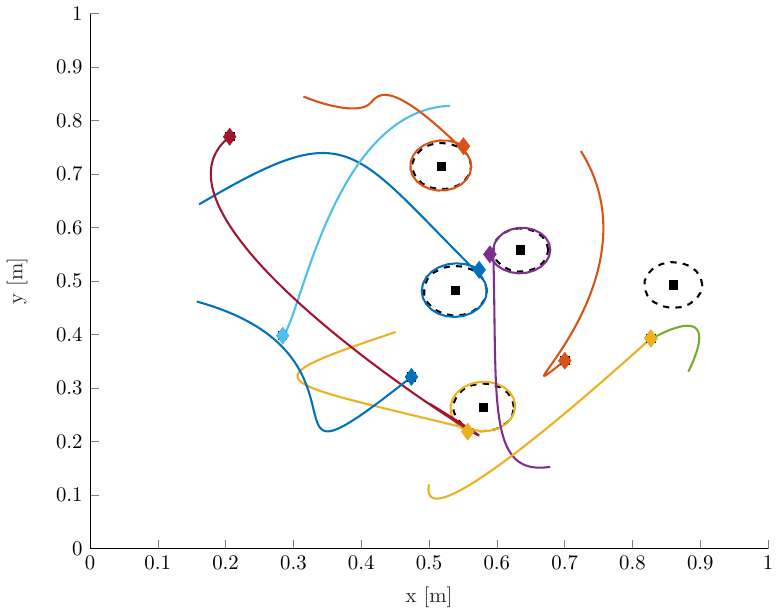}
  \caption{$t=10$}
  \label{fig:ExUC_Im3}
\end{subfigure}
\caption{Dynamic task allocation for the range unconstrained case ($n=p=10$, $\varrho \rightarrow \infty$, $t_\f=10$, $\cU=1.804$).}
\label{fig:RHTAUNC}
\end{figure}

While problems with zero terminal velocities have an analytical solution (\cite{p:bakolas2014}), problems with non-zero terminal velocities require more computation time due to the numerical integration.
The optimal cost-to-go is then obtained via the definition of $\cJ(u_i(\cdot))$.
In addition to this dynamic solution, a drag term (or friction force) $-k_\mathrm{d}\,\dot{p}_{\cT_j}$ proportional to the agent's velocity is used to refine the previous ideal equations of motion. It will thereby slow down to zero velocity an agent when he is not subject to any input control (i.e., he does not have any assigned task) while being negligible when the agent is subject to a typical control input.

\begin{figure}[h]
\centering
\begin{subfigure}{\linewidth}
  \centering
  \includegraphics[width=0.8\linewidth]{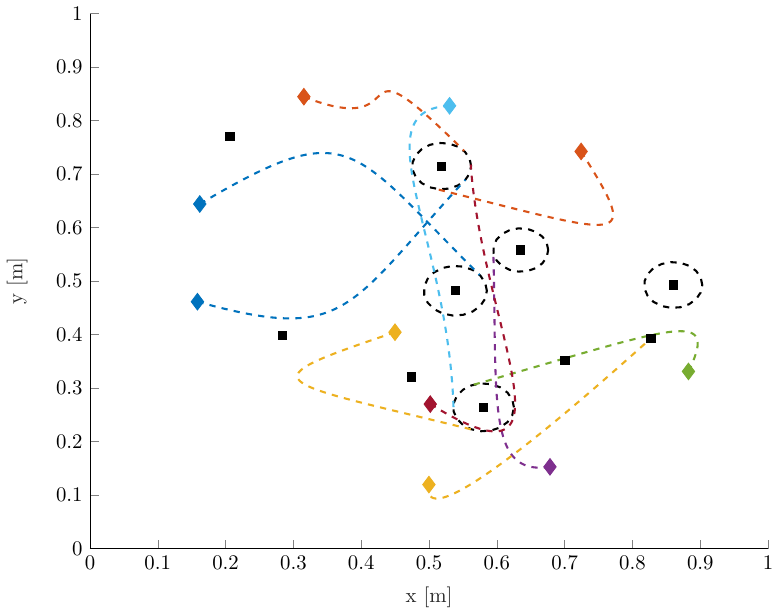}
  \caption{$t=0^+$}
  \label{fig:Ex_Im}
\end{subfigure}
\begin{subfigure}{\linewidth}
  \centering
  \includegraphics[width=0.8\linewidth]{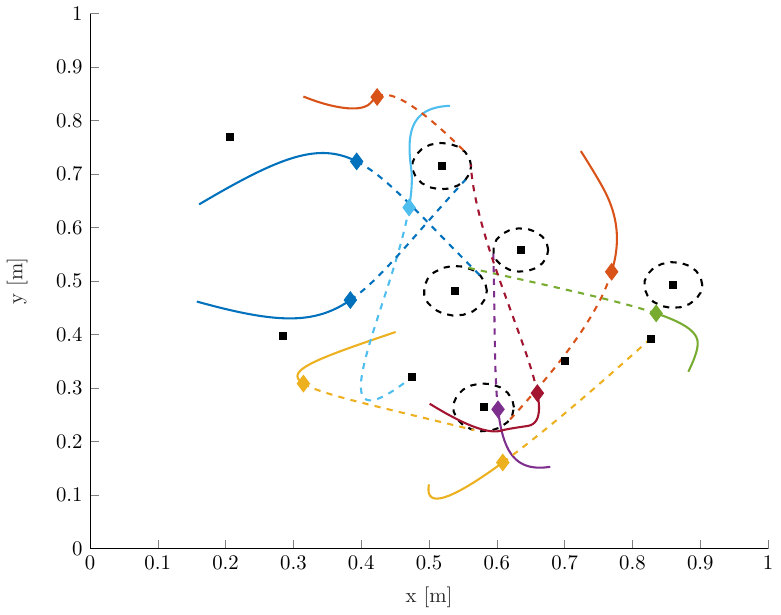}
  \caption{$t=4$}
  \label{fig:Ex_Im2}
\end{subfigure}
\begin{subfigure}{\linewidth}
  \centering
  \includegraphics[width=0.8\linewidth]{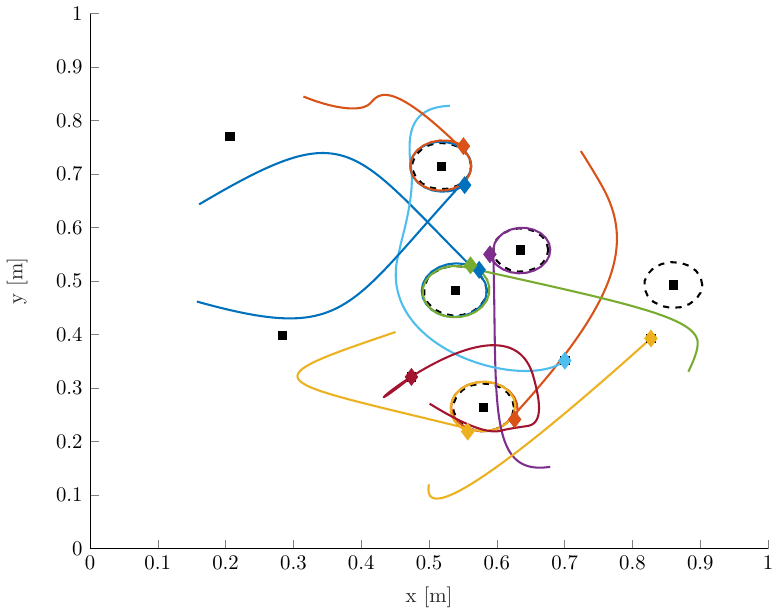}
  \caption{$t=10$}
  \label{fig:Ex_Im3}
\end{subfigure}
\caption{Dynamic task allocation for the range constrained case ($n=p=10$, $\varrho = 0.3$, $t_\f=10$, $\cU=1.515$)}
\label{fig:RHTA}
\end{figure}

A dynamic task allocation is then performed and presented through the dynamic map of the allocation. Fig.~\ref{fig:RHTAUNC} illustrates trajectories of the agents in the absence of communication constraints (or limitations) while the agents in Fig.~\ref{fig:RHTA} can only communicate\footnote{The communication range is not shown in the figure for clarity.} with the other agents within range $\varrho = 0.3$.
We set the simulation time to $t_\f=10$ and discretize it in $k=1000$ iterations which implies a constant time step $\delta t = t_\f/k = 0.01$.
Due to the fact that the computation time required for the numerical integration is substantial, we only consider scenarios with $n=10$ agents and $p=10$ tasks. The completion time $t_{\f,\cT_j}$, which is dependent on the task $\cT_j$, is computed randomly such that $t_{\f,\cT_j}/t_\f \in [0.9,1]$. 5 tasks are fixed targets with non-zero terminal velocities $\dot{p}_{\cT_j} \in [-0.1,0.1]$ (black squares). The other 5 tasks are dynamic, the agents need to loiter at a radius $\tilde{R}_{\cT_j} \in [0.032,0.048]$ and complete one loop at velocity $\dot{p}_{\cT_j} \in [-0.1,0.1]$ for a time $\tau_{\cT_j}$ such that $\tau_{\cT_j}/t_\f \in [0.15,0.25]$ (black dashed circle around a dashed square).

The time $t_{p}$ after which the algorithm preserves the same allocation for an agent is satisfying $\Phi_i( x_i(t_p), x_{\cT_j}) = \|x_i(t_p) - x_{\cT_j} \| - 2 \tilde{R}_{\cT_j} < 0 $ so that the allocation is blocked when the agent enters the circle of radius $2 \tilde{R}_{\cT_j}$ centered in $x_{\cT_j}$ before starting loitering.
The nominal rewards are such that $\bar{r}_{\cT_j} \in [0,0.2]$ for the fixed tasks, they are higher ($\bar{r}_{\cT_j} \in [0,1]$) for the loitering tasks since they typically require more cost to achieve the rotation. The success probability $p_{ij}$ is chosen randomly between 0 and 1.

As seen in Fig.~\ref{fig:RHTAUNC} for the range unconstrained case, the agents (colored diamonds) can freely communicate from start and thereby directly find the best allocation (dashed lines), which is maintained all along the trajectory (plain lines).

Conversely in Fig.~\ref{fig:RHTA} where the range is limited to 0.3, several agents are allocated to the same task because they are not in communication with all the other agents. They thus estimate their utility solely based on the reward of the task while their marginal utility is actually lower. When the agents come closer and enter in communication, the agents start assessing their marginal utility correctly and thus consider other tasks that might increase their own utility. Toward the end of the simulation (Fig.~\ref{fig:Ex_Im3}), the agents' trajectory is subject to sharper changes of direction (e.g. the red and light blue curves) compared to the range unconstrained case (Fig.~\ref{fig:ExUC_Im3}).

\begin{figure}[h]
  \centering
  \includegraphics[width=0.8\linewidth]{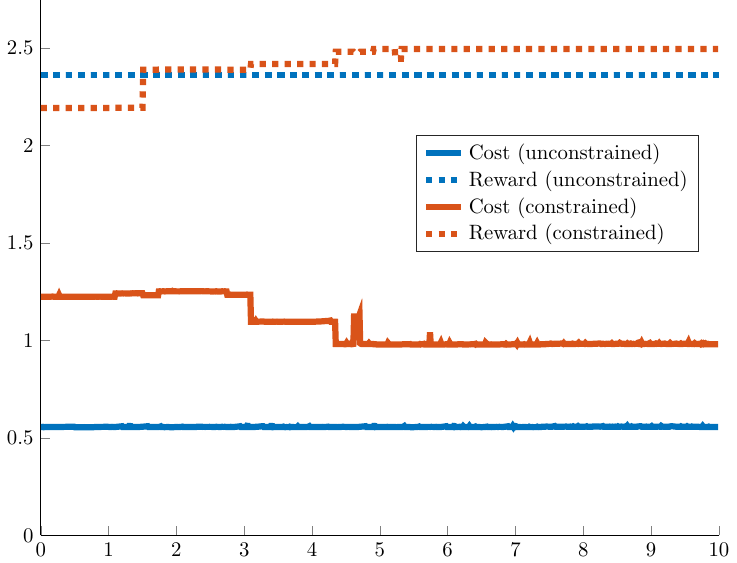}
  \caption{Total cost and reward.}
  \label{fig:CostReward}
\end{figure}

Fig.~\ref{fig:CostReward} and \ref{fig:TotalUtility} quantitatively show the allocation presented above, for which the data from the range unconstrained case (blue lines) are constant over time. When the communication is limited, the total utility increases step-by-step as the agents start communicating with their neighbors (red line in Fig.~\ref{fig:TotalUtility}) but still remains lower than the utility obtained when the communication is not limited. It is worth noting that even though the final reward is higher for the constrained case (dashed red line in Fig.~\ref{fig:CostReward}), the higher cost produced by abrupt trajectory changes makes its final utility lower than the utility generated without communication limitation. Finally, the noisy curves are due to approximation errors in the numerical integration.

\begin{figure}[h]
  \centering
  \includegraphics[width=0.8\linewidth]{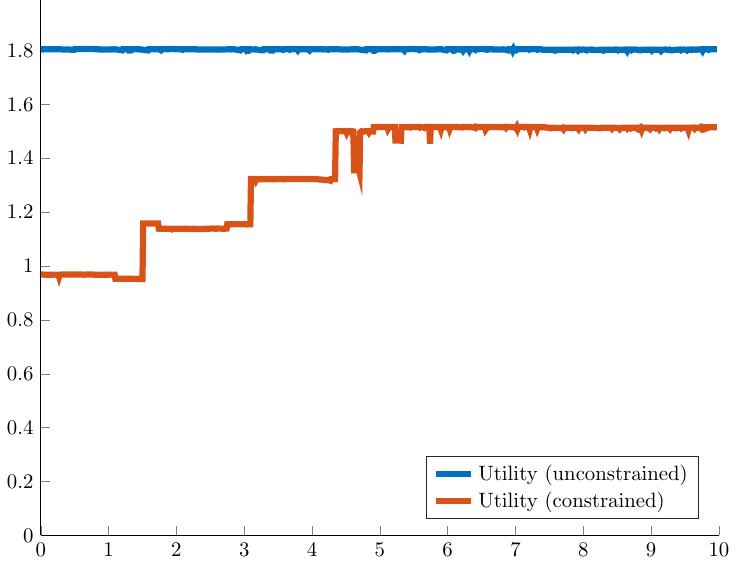}
  \caption{Total utility.}
  \label{fig:TotalUtility}
\end{figure}

The impact of several parameters on the utility and the computation time is then performed (with 10 agents and 10 tasks if not mentioned). In Fig.~\ref{fig:RangeLimitationUtility}, the global utility increases progressively with the communication range. 

\begin{figure}[h]
  \centering
  \includegraphics[width=0.8\linewidth]{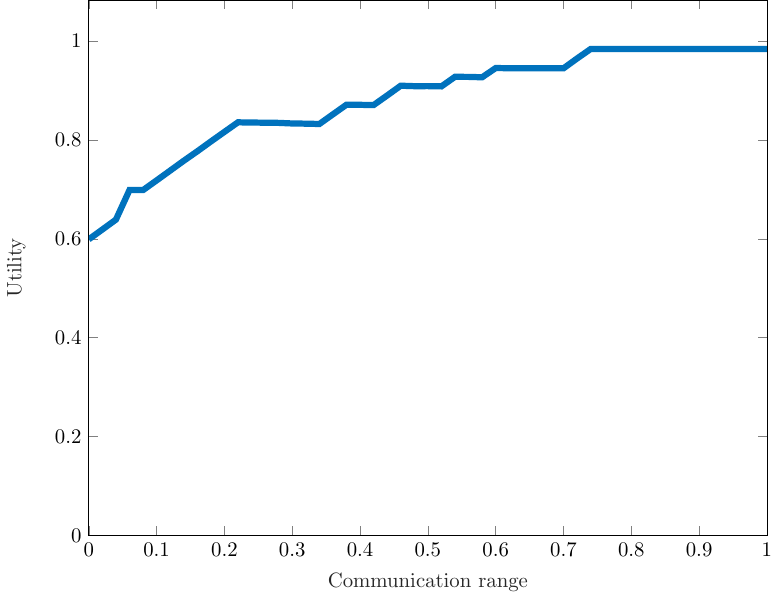}
  \caption{Impact analysis of the communication range on the utility.}
  \label{fig:RangeLimitationUtility}
\end{figure}

\begin{figure}[h]
  \centering
  \includegraphics[width=0.8\linewidth]{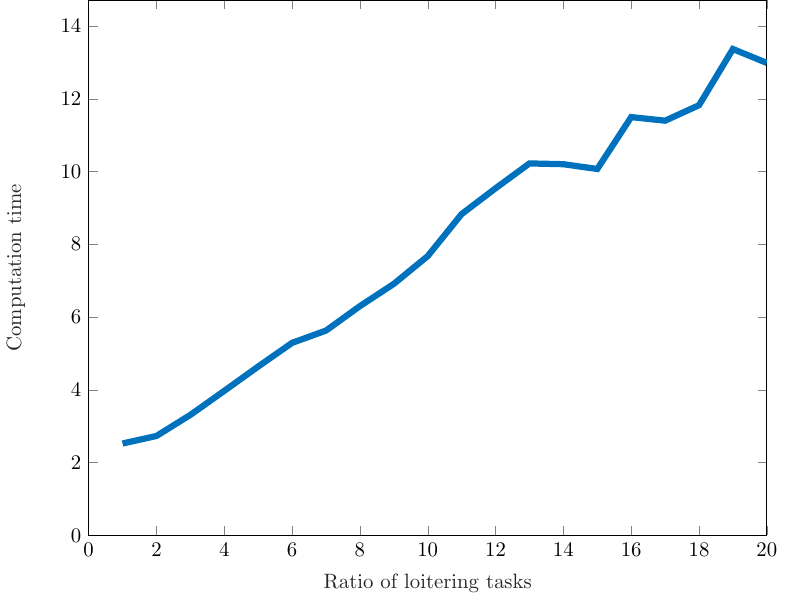}
  \caption{Impact analysis of the ratio of loitering tasks on the computation time.}
  \label{fig:ntLoiterComputationTime}
\end{figure}

In Fig.~\ref{fig:ntLoiterComputationTime}, we perform a simulation with 20 tasks and we progressively replace the tasks at a fixed location with loitering tasks. Since the number of numerical integrations is ten times higher for the latter, the computation time for 20 fixed tasks is approximately ten times higher than the one for 20 loitering tasks.

The effect of the number of agents / tasks on the computational cost is analyzed separately in Fig.~\ref{fig:nTnAUtility}. For a fixed number of agents, the global utility is linear with the number of tasks. It is however interesting to note that for a fixed number of tasks (e.g., 80 at the graph boundary), the utility increases exponentially with the number of agents. This shows that in this framework, multiplying the fleet of agents by $m \in \mathbb{R}$ will result in a global utility lower than $m\,\cU$. This can be illustrated by considering one task $\cT_1$ and $n$ agents with probability $p$ to complete the task, then the utility is given by $\bar{r}_{\cT_1}[1-(1-p)^n]$. 

\begin{figure}[h]
  \centering
  \includegraphics[width=0.7\linewidth]{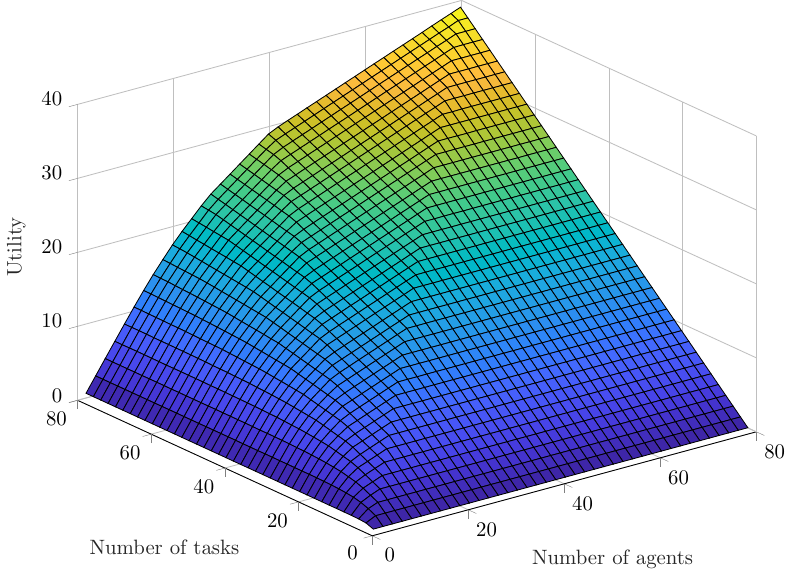}
  \caption{Impact analysis of the number of tasks / agents on the global utility.}
  \label{fig:nTnAUtility}
\end{figure}

Furthermore in Fig.~\ref{fig:nTnAComputationTime}, the computation time is more dependent on the number of tasks than the number of agents (the allocation with 1 task and 50 agents is straightforward while the allocation with 50 tasks and 1 agent requires the agent to iterate over all tasks). As a consequence, problems with a large number of agents are more tractable than problems with a high number of tasks. To mitigate this issue, one could for example restrict the considered tasks to the tasks near the agents (but at the cost of the utility, revealing a trade-off between the computation and the optimal allocation with the maximum utility).

\begin{figure}[h]
  \centering
  \includegraphics[width=0.7\linewidth]{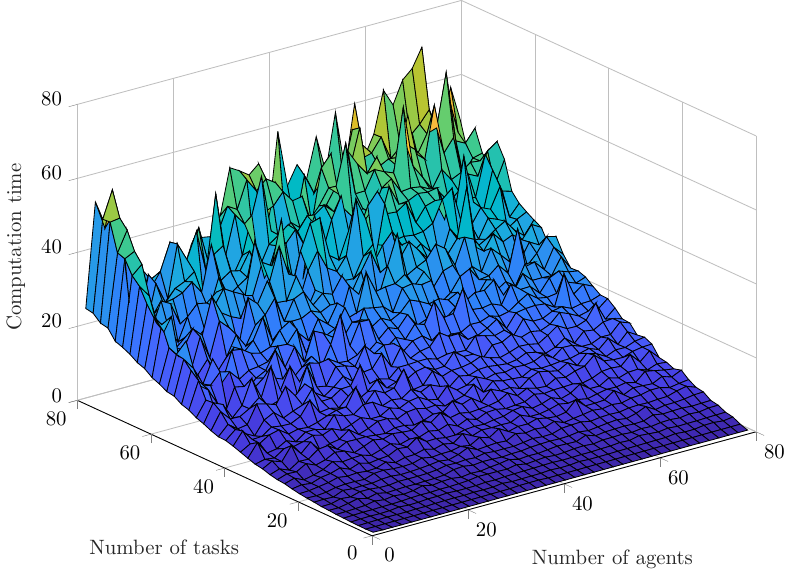}
  \caption{Impact analysis of the number of tasks / agents on the computation time.}
  \label{fig:nTnAComputationTime}
\end{figure}

\section{Conclusion}\label{s:concl}

In this paper, we have presented an auction-based framework to address dynamic task allocation problems for multi-agent systems
with state-dependent utilities and various task characteristics (such  as terminal constraints, completion time, etc.). Our greedy approach offers a practical, yet efficient, solution to a class of more realistic and challenging dynamic task allocation problems for autonomous mobile agents.

For large fleets of autonomous systems, some scalability issues arise due to the computation time. Our next research will focus more on machine learning and its recent achievements (specifically deep reinforcement learning for multi-agent robots) to mitigate this issue and build upon the algorithm presented in this work. We further plan to extend the results presented herein to even more realistic task allocation problems with deadlines, logical constraints, pop-up tasks and agents with varying capabilities and preferences. Finally, the control problem will be extended to include obstacles (e.g., cluttered environments).

\bibliography{ref}

\begin{thebibliography}{24}
\providecommand{\natexlab}[1]{#1}
\providecommand{\url}[1]{\texttt{#1}}
\providecommand{\urlprefix}{URL }
\expandafter\ifx\csname urlstyle\endcsname\relax
  \providecommand{\doi}[1]{doi:\discretionary{}{}{}#1}\else
  \providecommand{\doi}{doi:\discretionary{}{}{}\begingroup
  \urlstyle{rm}\Url}\fi

\bibitem[{Arslan et~al.(2007)Arslan, Marden, and Shamma}]{p:shamma2007}
Arslan, G., Marden, J.R., and Shamma, J.S. (2007).
\newblock {Autonomous Vehicle-Target Assignment: A Game-Theoretical
  Formulation}.
\newblock \emph{J. Dyn. Syst. Meas. Control}, 129(5), 584--596.

\bibitem[{{Bakolas}(2014)}]{p:bakolas2014}
{Bakolas}, E. (2014).
\newblock A decentralized spatial partitioning algorithm based on the minimum
  control effort metric.
\newblock \emph{2014 American Control Conference}, 5264--5269.

\bibitem[{Bakolas and Lee(2021)}]{p:bakolas2020}
Bakolas, E. and Lee, Y. (2021).
\newblock Decentralized game-theoretic control for dynamic task allocation
  problems for multi-agent systems.
\newblock In \emph{2021 American Control Conference (ACC)}.

\bibitem[{Bakshi et~al.(2019)Bakshi, Feng, Yan, and Chen}]{p:bakshi2019}
Bakshi, S., Feng, T., Yan, Z., and Chen, D. (2019).
\newblock Fast scheduling of autonomous mobile robots under task space
  constraints with priorities.
\newblock \emph{ASME. J. Dyn. Sys., Meas., Control.}

\bibitem[{{Bakshi} et~al.(2019){Bakshi}, {Feng}, {Yan}, and
  {Chen}}]{p:bakshi2019quad}
{Bakshi}, S., {Feng}, T., {Yan}, Z., and {Chen}, D. (2019).
\newblock A regularized quadratic programming approach to real-time scheduling
  of autonomous mobile robots in a prioritized task space.
\newblock \emph{2019 American Control Conference (ACC)}, 1361--1366.

\bibitem[{Battin(1987)}]{p:battin1987}
Battin, R. (1987).
\newblock An introduction to the mathematics and methods of astrodynamics.
\newblock 559--561.

\bibitem[{Capitan et~al.(2013)Capitan, Spaan, Merino, and
  Ollero}]{p:capitan2013}
Capitan, J., Spaan, M.T., Merino, L., and Ollero, A. (2013).
\newblock Decentralized multi-robot cooperation with auctioned \textsc{POMDP}s.
\newblock \emph{The International Journal of Robotics Research}, 32(6),
  650--671.

\bibitem[{{Choi} et~al.(2009){Choi}, {Brunet}, and {How}}]{p:howauction2009}
{Choi}, H., {Brunet}, L., and {How}, J.P. (2009).
\newblock Consensus-based decentralized auctions for robust task allocation.
\newblock \emph{IEEE Transactions on Robotics}, 25(4), 912--926.

\bibitem[{Gautier et~al.(2020)Gautier, Johann, and Diguet}]{p:gautier2020}
Gautier, P., Johann, L.D., and Diguet, J.P. (2020).
\newblock {Comparison of Market-based and DQN methods for Multi-Robot
  processing Task Allocation (MRpTA)}.
\newblock \emph{{IEEE International Conference on Robotic Computing (IRC)}}.

\bibitem[{Gerkey and Mataric(2002)}]{p:gerkey2002}
Gerkey, B.P. and Mataric, M.J. (2002).
\newblock Sold!: Auction methods for multirobot coordination.
\newblock \emph{IEEE Transactions on Robotics and Automation}, 18(5), 758--768.

\bibitem[{Gerkey and Matarić(2004)}]{p:gerkey2004}
Gerkey, B.P. and Matarić, M.J. (2004).
\newblock A formal analysis and taxonomy of task allocation in multi-robot
  systems.
\newblock \emph{The International Journal of Robotics Research}, 23(9),
  939--954.

\bibitem[{Hayat et~al.(2020)Hayat, Yanmaz, Bettstetter, and
  Brown}]{p:hayat2020}
Hayat, S., Yanmaz, E., Bettstetter, C., and Brown, T.X. (2020).
\newblock Multi-objective drone path planning for search and rescue with
  quality-of-service requirements.
\newblock \emph{Autonomous Robots}, 44(7), 1183--1198.

\bibitem[{Johnson et~al.(2011)Johnson, Ponda, Choi, and How}]{p:johnson2011}
Johnson, L.B., Ponda, S.S., Choi, H., and How, J. (2011).
\newblock Asynchronous decentralized task allocation for dynamic environments.
\newblock \emph{Infotech@Aerospace 2011}.

\bibitem[{Kim et~al.(2019)Kim, Kim, and Choi}]{p:keum2019}
Kim, K.S., Kim, H.Y., and Choi, H.L. (2019).
\newblock Minimizing communications in decentralized greedy task allocation.
\newblock \emph{Journal of Aerospace Information Systems}, 16, 1--6.

\bibitem[{Lee(2018)}]{p:lee2018}
Lee, D.H. (2018).
\newblock Resource-based task allocation for multi-robot systems.
\newblock \emph{Robotics and Autonomous Systems}, 103, 151--161.

\bibitem[{{Luo} et~al.(2012){Luo}, {Chakraborty}, and {Sycara}}]{p:luo2012}
{Luo}, L., {Chakraborty}, N., and {Sycara}, K. (2012).
\newblock Competitive analysis of repeated greedy auction algorithm for online
  multi-robot task assignment.
\newblock In \emph{2012 IEEE International Conference on Robotics and
  Automation}, 4792--4799.

\bibitem[{Nanjanath and Gini(2010)}]{p:NANJANATH2010}
Nanjanath, M. and Gini, M. (2010).
\newblock Repeated auctions for robust task execution by a robot team.
\newblock \emph{Robotics and Autonomous Systems}, 58(7), 900 -- 909.

\bibitem[{Otte et~al.(2020)Otte, Kuhlman, and Sofge}]{p:otte2020}
Otte, M., Kuhlman, M.J., and Sofge, D. (2020).
\newblock Auctions for multi-robot task allocation in communication limited
  environments.
\newblock \emph{Autonomous Robots}, 44(3), 547--584.

\bibitem[{Phillips and Parra(2021)}]{p:phillips2021}
Phillips, S. and Parra, F. (2021).
\newblock A case study on auction-based task allocation algorithms in
  multi-satellite systems.
\newblock \emph{AIAA Scitech 2021 Forum}.

\bibitem[{Qu et~al.(2019)Qu, Brown, and Li}]{p:guannan2019}
Qu, G., Brown, D., and Li, N. (2019).
\newblock Distributed greedy algorithm for multi-agent task assignment problem
  with submodular utility functions.
\newblock \emph{Automatica}, 105, 206 -- 215.

\bibitem[{Rahili et~al.(2020)Rahili, Riviere, and Chung}]{p:rahili2020}
Rahili, S., Riviere, B., and Chung, S.J. (2020).
\newblock Distributed adaptive reinforcement learning: A method for optimal
  routing.
\newblock \emph{arXiv preprint arXiv:2005.01976}.

\bibitem[{Ravichandar et~al.(2019)Ravichandar, Shaw, and
  Chernova}]{p:ravichandar2019}
Ravichandar, H., Shaw, K., and Chernova, S. (2019).
\newblock Strata: A unified framework for task assignments in large teams of
  heterogeneous robots.
\newblock \emph{Autonomous Agents and Multi-Agent Systems}, 34, 38.

\bibitem[{Shin et~al.(2019)Shin, Li, and Segui-Gasco}]{p:shin2019}
Shin, H.S., Li, T., and Segui-Gasco, P. (2019).
\newblock Sample greedy based task allocation for multiple robot systems.
\newblock \emph{arXiv preprint arXiv:1901.03258}.

\bibitem[{Whitbrook et~al.(2019)Whitbrook, Meng, and Chung}]{p:whitbrook2019}
Whitbrook, A., Meng, Q., and Chung, P.W.H. (2019).
\newblock Addressing robustness in time-critical, distributed, task allocation
  algorithms.
\newblock \emph{Applied intelligence}, 49(1), 1--15.

\end{thebibliography}

\end{document}